\documentclass[12pt,oneside,reqno]{amsart}

\usepackage{amssymb,amsfonts}
\usepackage[all,arc]{xy}
\usepackage{enumerate}
\usepackage{mathrsfs}

\newtheorem{thm}{Theorem}[section]

\newtheorem{prop}[thm]{Proposition}
\newtheorem{lem}[thm]{Lemma}

\newtheorem{assum}{Assumption}

\theoremstyle{definition}

\theoremstyle{remark}
\newtheorem{rem}[thm]{Remark}

\makeatletter
\def\Let@{\def\\{\notag\math@cr}}
\makeatother
\numberwithin{equation}{section}

\title{Portfolio Selection with Mandatory Bequest}

\newcommand*{\bigchi}{\mbox{\large$\chi$}}

\author{Jiacheng Feng}

\date{September 1, 2014}

\begin{document}
\maketitle
\begin{center} \footnotesize 
Mentor: Menglu Wang \\
\vspace{4 mm}
UROP+ Final Paper, Summer 2014
\end{center}

\vspace{10 mm}

\begin{abstract}
In this paper, optimal consumption and investment decisions are studied for an investor who can invest in a fixed interest rate bank account and a stock whose price is a log normal diffusion. We present the method of the HJB equation in order to explicitly solve problems of this type with modifications such as a fixed percentage transaction cost and a mandatory bequest function. It is shown that the investor treats the mandatory bequest as an expense that she factors into her personal wealth when making consumption and transaction decisions. Furthermore, the investor keeps her portfolio proportions inside a fixed boundary relating to Merton's optimal proportion and the transaction costs. 

\end{abstract}

\clearpage
\section{Introduction}
A number of investigators have studied the optimal consumption-portfolio policy in continuous time by the method of Stochastic Dynamic programming.  Pioneered by Merton [1], the primary method in solving these types of consumption-portfolio problems is to show the existence of the solution using stochastic control. Most of these portfolio selection works focus on the intrinsic workings of the model and the utility rate function derived from consumption, but dismiss the bequest function by assuming that it is $0$. However, the bequest function is often very relevant in that it demonstrates different possible termination conditions of the control problem, allowing for a more dynamic model. 

The objective of this paper is first to introduce the theory behind solving stochastic control problems and then to analyze a very interesting bequest function, similar to an indicator function, for which the investor must hold on to at least a specific amount of assets.  It is found that when this compulsory goal is added, the investor will remove this fixed wealth from her effective total wealth and adjust her consumption and transaction policies accordingly. Qualitatively, with a higher mandatory bequest, optimal consumption decreases and optimal portfolio proportions will shift towards the safe asset. This problem can find applications in outperforming benchmarks in the financial world as well as understanding money management in survival of low income families,  however, these applications are not dealt with in this paper. A presentation on the effects of a fixed percentage transaction cost is also exhibited, in which we find that the investor will not trade in securities if and only if  her portfolio proportions remain inside a certain region about the zero transaction cost optimal proportions.

The rest of the paper is organized as follows. Section 2 contains general theory behind the optimal portfolio control problem. A heuristic proof of the main theorem is provided.  Section 3 displays explicit underlying assumptions about the market, as well as the behavior of the investor, in order to formalize the portfolio problems. Sections 4 through 7 present different cases of optimal consumption with increasing difficulty by carefully and explicitly deriving the optimality equations for the two-asset problem where the rate of returns are generated by It\^o stochastic processes under isoelastic marginal utility. These problems will build up the necessary conceptual ideas used to solve the problem that requires the investor to satisfy a mandatory bequest under transaction costs. The paper ends with concluding remarks, as well as discussions on work that could be expanded upon in the future. 

{\small \subsection*{Acknowledgments}  It is my pleasure to thank my mentor, Menglu Wang, for her immense help in guiding me through the entire process of research and writing. Of course, all errors in this paper are mine. I would like to also thank Prof. Scott Sheffield for his valuable discussions and suggestions. I thank Prof. Pavel Etingof and the MIT Department of Mathematics for directing the Math UROP+ program, and I gratefully acknowledge the aid from the Paul E. Gray (1954) Endowed Fund for UROP.}

\section{Stochastic Control and the HJB Equation}
Suppose that any instant $t$ the stochastic process  $X_t \in \mathbb{R}^n$ defined on a probability space $(\Omega, \mathcal{F}, P)$ can be influenced by a choice of a parameter $u_t \in U \subset \mathbb{R}^k$, which is called the \emph{control}. Here, assume that $u_t$ only depends on the current state of the system at the time, that is,  $u_t=u(t,X_t)$ is a \emph{Markov control}. Since $u_t$ is determined only by what is happening at time $t$, the function $\omega \rightarrow u(t,X_t(w))$ must be measurable with respect to the filtration $\mathcal{F}_t$, thus the process $u_t$ is $\mathcal{F}_t $-adapted stochastic process. \\

Let the system $X_t$ be described by the well defined stochastic differential equation with initial value:
\begin{align}
dX_t = dX_t^u &= b(X_t, t, u_t)dt + \sigma(X_t, t, u_t)dB_t \; \text{for } t \in (s, T] \\
X_0 &= x_0
\end{align}
where $b: \mathbb{R}^n \times  \mathbb{R} \times U \rightarrow  \mathbb{R}^n$, $\sigma: \mathbb{R}^n \times  \mathbb{R} \times U \rightarrow  \mathbb{R}^{n \times m}$, and $B_t$ the $m$-dimensional Brownian motion.The objective is to set the control $u$ to maximize the \emph{performance function} $J^u (s,x)$,defined as 
\begin{equation} 
J^u(t,x) = \mathbb{E}^{t,x} \left[ \int ^T_t  f^u(s, X_s) ds + g(T, X_T) \right] 
\end{equation}

We can view $f: \mathbb{R}^n \times  \mathbb{R} \times U \rightarrow  \mathbb{R}$ to be the profit rate function, $g: \mathbb{R}^n \times  \mathbb{R} \rightarrow \mathbb{R}$ to be the bequest function, and $T$ to be the first exit time from a solvency set $G$ ($G$ could be the whole space). Assume that $f, g$ are continuous, and $U$ compact. The main question is, for each $(t, X_t)$ can we find an \emph{optimal control} $u^* = u^*(t, X_t)$ and its corresponding \emph{optimal performance function} $\phi(t, X_t)$ such that 
$$\phi(t, X_t) = \underset{u(t,X_t)}{\sup} J^u(t, X_t) = J^{u^*}(t, X_t) ?$$

We introduce the concept of the \textbf{Hamilton-Jacobi-Bellman (HJB) equation} [2], which provides the optimal performance function as solution to the continuous time optimization problems.  First, define the differential operator $\mathcal{L}^v$ to be 
\begin{equation} (\mathcal{L}^v f)  = \frac{\partial f}{\partial t} (t,x) + \sum\limits_{i=1}^n b_i(t,x,v) \frac{\partial f}{\partial x_i} + \sum\limits_{i, j =1}^n  a_{ij} (t,x,v) \frac{\partial^2 f}{\partial x_i \partial x_j} \end{equation}
where $a_{ij} = \frac 12 {\sigma \sigma^T}_{ij}$ and $x = (x_1, ..., x_n)$.

With the same notations as above in the problem statement of optimal control, 
\begin{thm}[HJB equation]
Let 
$$\phi(s,x) = \sup \{ J^u(s,x); u = u(s+t, X_{s+t} ) \}$$
Suppose $\phi \in C^2(G) \cap C( \bar{G} )$ satisfies $$\mathbb{E}[ | \phi ( \alpha, X_{\alpha} ) | + \int_0^\alpha | \mathcal{L}^v \phi(t, X_t) | dt ] < \infty$$
for all bounded stopping times $\alpha \leq T$, all states $(t, X_t) \in G$, and all control $v \in U$. Furthermore, suppose that an optimal Markov control $u^*$ exists, then
\begin{equation} \underset{v \in U}{\sup} \{ f^v(t, X_t) + (\mathcal{ L}^v \phi) (t, X_t) \} = 0 \qquad \text{for all } (t, X_t) \in G \end{equation}
and 
\begin{equation} \phi(t, X_t) = g(t, X_t) \qquad \text{for all } (t, X_t)  \in \partial G \end{equation}
The supremum is obtained when $v = u^*$, the optimal control. \\

Conversely, let $\phi$ be a function in $C^2(G) \cap C( \bar{G} )$ with boundary condition $\underset{t \rightarrow T}{\lim} \phi(t,X_t) = g(T,X_T) \cdot \bigchi _{\{ T < \infty \} }$. Suppose for all control $v \in U$ and all states $(t, X_t) \in G$,
$$f^v(t, X_t) + (L^v \phi) (t, X_t) \leq 0,$$ and $\phi$ is uniformly integrable with respect to the measure.
Then,
$$ \phi(t,X_t) \geq J^u (t, X_t)$$
for all control $u \in U$ and all states $(t, X_t) \in G$. Furthermore, if there exists a control $v$ such that
\begin{equation} f^{v}(t, X_t) + (\mathcal{L}^{v} \phi) (t, X_t) = 0 \end{equation}
then $v = v(t,X_t) = u^*$ is an optimal control that satisfies 
\begin{equation} \phi(t,X_t) = J^{v}(t, X_t)= \underset{u \in U}{\sup} \{J^u(t,X_t) \} \end{equation}

\end{thm}

\begin{rem}
The \textbf{Bellman equation} is developed for mathematical optimization in order to solve the problem of maximizing utility subject to a budget constraint in discrete time with intervals $\delta$. The idea behind the equation is to forcefully find an optimal policy with the property that regardless of the initial states and decisions, the remaining decisions be optimal with regards to the current state. With this central idea, the Bellman equation states that the value function $\phi$ in discrete time must satisfy
\begin{equation} \phi(t, x_t) = \underset {u \in U}{\max} \{  \int f^u(t, X_t) dt + \phi(R(X_t,u) \} \end{equation}
where $R(x_t,u)$ denotes the state change at time $t + \delta$ of $x$ after applying the control $u$. The HJB equation is derived by extending the discrete time Bellman equation with the Hamilton-Jacobi equation in physics.
\end{rem}

\begin{proof}
An intuitive sketch of the HJB equation using the ideas of the Bellman equation is derived [4]. For a complete rigorous proof, see {\O}ksendal [2]. 
Let $$\phi(s,x) = \sup \{ J^u(s,x) \}$$ and $$J^u(s,x) = \mathbb{E}^{t,x} \left[ \int ^T_t  f^u(s, X_s) ds + g(T, X_T) \right]$$ where the supremum is taken over all possible controls $u$ starting from $(s,x)$. Suppose that the control $v$ is chosen for the time interval $(t, t+\delta)$ before switching to the optimal control $u$. Then, comparing this with the established optimal control, it must be that
\begin{equation} \phi(t,x(t)) \geq f^v(t,x) \delta + \phi(t+\delta, R(x, v)) \end{equation}
The Taylor expansion on $ \phi (t+\delta, R(x, v))$ is 
\begin{equation} \phi(t+\delta, T(x, v))  = \phi(t,x_t) + \frac{\partial \phi}{\partial t}(t, x_t) \delta + \nabla \phi(t, x_t) \cdot x'_t \delta + o(\delta) \end{equation}
where $\nabla$ is the Laplacian operator with respect to $x$, and $o(\delta)$ is the terms in the Taylor expansion with order greater than one. If we cancel $\phi(t, x_t)$ from both sides and divide by $\delta$ in $(2.10)$, and take $\delta \rightarrow 0$ such that $o(\delta) \rightarrow 0$, then
\begin{equation} 0 \geq f^v(t,x) \delta + \mathcal{L} \phi(t, x_t) \end{equation}
Furthermore, if $v$ is the optimal policy, in that $v = u$, then equality must hold.  
\end{proof}

\section{Assumptions of the Model}
Suppose the investor faces a capital market with the following properties: 

\begin{assum}\label{as:1}
Securities and Market
\end{assum}

We assume that the market in question is perfectly competitive, and trading takes place in continuous time. There are two underlying securities which can be bought and sold at current prices for unlimited amounts. The prices of these two securities $\{P_i(t)\}$ can be identified by It\^o stochastic differential equations: 
$$\frac{dP_i}{dt} = \alpha_i(t,P_i) dt + \sigma(t,P_i) dB_t$$
where $\alpha_i$is the expected value of the percent change in price, $\sigma_i^2$ the variance of this change, and $B_t$ a 1-dimensional Brownian Motion.

The first (bank) security is a safe investment with unit price $\{P_0(t)\}$ pays a constant risk free interest rate $r_0> 0$ for all investments, and charges the same rate on borrowing. The value of the bank securities does not exhibit inflation or deflation.The bank security satisfies 
\begin{equation} dP_0(t) = P_0(t) r_0 dt \end{equation}

The second (stock) security is a risky investment with unit price at time $t$ being $\{P_1(t) \}$, which satisfies the equation
$$\frac{dP_1(t)}{dt} =  P_1(t)[r_1 + s_1 W_t] dt$$
where $W_t$ denotes white noise and $r_1, s_1$ are constants measuring average rate of change and size of the noise. Then,the price function under It\^o stochastic differential equation can be formulated as[2]:
\begin{equation} dP_1(t) = P_1(t) r_1 dt + P_1(t) s_1 dB_t \end{equation}

Let both securities be perfectly divisible. It is natural to assume that the risky security would have a higher expected rate of return than the safe one, thus let $r_1 > r_0 > 0$.

\begin{assum}\label{as:2}
Information
\end{assum}
The probability distribution and the current price of the underlying securities contain all necessary information for any investor to make her decision. This information is publicly and continuously available to all investors free of cost. 

\begin{assum}\label{as:3}
Transaction Costs
\end{assum}
In the third and fourth problem, transaction fees will be incurred for buying or selling stocks. If $v$ denotes the value of the \emph{risky} security that is bought ($v> 0$) or sold ($v < 0$), the cost of transaction per unit is described as a premium cost $\bigchi_0$ that measures whether cash or stock more desirable plus a fee $\bigchi_v$ proportional to the value of the transaction. Concretely, the transaction cost $\tau$ can be written as 
 \begin{displaymath}
   \tau(v) = v(\bigchi_v + \bigchi_0) = \left\{
     \begin{array}{lr}
       |v|(\bigchi+ \bigchi_0) & : v > 0\\
       -|v|(-\bigchi+ \bigchi_0) & : v < 0 \\
       0 &: v = 0
     \end{array}
   \right.
\end{displaymath} 
where $\bigchi$, $\bigchi_0$  are reasonable small constants in $[0,1)$  and $(-1, 1)$ respectively. Then, the per unit transaction costs are $\bigchi+ \bigchi_0$ and $\bigchi - \bigchi_0$ as $v > 0$ and $v < 0$ respectively.

\begin{assum}\label{as:4}
Income and Lifespan
\end{assum}

The investor has a lifespan from $[0,T]$, during which she is expected to earn an influx of income $y(t)$ per unit of time. Assume that $y(t)$ is integrable from any interval in $[0,T]$.
At $t=0$, she begins with an initial fortune $Z_0$, which is all put in the safe asset. The investor acts as if both $T$, $y(t)$, and $Z_0$ are known with certainty at any $t \in [0,T]$. Denote $Z(t) = Z_t$ is the total wealth on the individual. 
Before going into formulating the stochastic differential equation for wealth, we first state and prove the following lemma.
\begin{lem}[Stochastic Integration by Parts] Let $ X_t, Y_t$ be It\^o processes in $\mathbb{R}^n$. Then
$$d(X_tY_t) = X_tdY_t + Y_t dX_t + dX_t \cdot dY_t$$
\end{lem}
\begin{proof} Applying It\^o's formula[5] with $g(x,y) = x \cdot y$ it is easy to obtain that 
\begin{align*}
d(X_tY_t) =& d(g(X_t, Y_t)) \\
=& \frac{\partial g}{\partial x}(X_t, Y_t) dX_t + \frac{\partial g}{\partial y}(X_t, Y_t) dY_t+  \frac 12 \frac{\partial^2 g}{\partial x^2}(X_t, Y_t) (dX_t)^2\\
&+  \frac{\partial^2 g}{\partial x \partial y}(X_t, Y_t) dX_t dY_t + \frac 12 \frac{\partial^2 g}{\partial y^2}(X_t, Y_t) (dY_t)^2\\
=& Y_tdX_t + X_tdY_t + dX_t dY_t
\end{align*}
\end{proof}

Let $N_1(t)$ denote the number of shares held in the risky security, and $N_0(t)$ denote the amount of money held in the bank, both at time $t$. We first consider a discrete time frame with fixed interval $h$ for the wealth function:
\begin{equation} Z(t) = N_0(t)P_0(t) + N_1(t)P_1(t) = \sum N_i(t)P_i(t) \end{equation}
where $P_i$ are the prices of the securities as described in Assumption 1. The summation $\sum = \sum_{i=0}^{1}$ is written for simplicity sake. Using Lemma 3.1 on $N_0(t)P_0(t)$ and $ N(t)P_1(t)$, and letting $h \rightarrow 0$, we obtain
\begin{equation} dZ_t =\sum (N_i(t)dP_i(t) + P_i(t)dN_i(t) + dN_i(t)dP_i(t)) \end{equation}
Any changes in the number of shares $N_i$ does not affect the total wealth since every sale of asset is accompanied by the buying of the other of the same worth. Then, if $ D(t)$ is the total adjustment of wealth from $t-h$ to $t$,
$$ D(t)h = \sum[N(t)-N(t-h)]P_i(t)$$
The next interval of wealth adjustment at $t+h$ is similarly
\begin{align*}D&(t+h)h = \sum[ N_i(t+h)-N_i(t)]P_i(t+h)\\
 =& \sum[N_i(t+h) - N_i(t)][P_i(t+h) - P_i(t)] + \sum[N_i(t+h) - N_i(t)]P_i(t)
\end{align*}
Taking the limit of $h \rightarrow 0$, the continuous version of changes in wealth can be formulated.
\begin{equation} D(t)dt =  \sum dN_i(t)P_i(t) + dN_i(t)dP_i(t) \end{equation}
Substituting $D(t)$ into $(3.4)$, the equation for $dZ_t$,
$$dZ_t = N_0(t)dP_0(t) + N_1(t)dP_1(t) + D(t)dt$$
Let $u$ be the proportion of the wealth invested in the risky investment so that
$$u = N_1(t)P_1(t)/Z(t) \qquad (1-u) = N_0(t)P_0(t)/Z(t)$$
$$dZ_t = Z(t)[(1-u) dP_0(t)/P_0(t) + u dP_1(t)/P_1(t)] + D(t)dt $$
Substituting  $dP_i(t)/P_i(t) $ for the stochastic formulation of prices in Assumption 1, the continuous stochastic differential equation for the wealth can be formulated as
\begin{equation} dZ_t = dZ_t^{(u)} = Z_t[r_0(1-u(t))+r_1u(t)]dt + Z_t[s_1u(t)]dB_t +  D(t)dt \end{equation}

\begin{assum}\label{as:5}
Utility Function
\end{assum}

We pick the isoelastic utility function $$U(t, x) = \frac {s^{\gamma}}{\gamma}$$
and let $\gamma$ be a chosen constant in $(0,1)$ so that the utility function is concave and satisfies basic utility function properties. Arrow Pratt's measure of relative risk aversion[6] $R(c) = -U''(c)/U'(c) = 1-\gamma$ is constant. Then, this utility function is a member of the family of utilities characterizing an investor with Constant Relative Risk Aversion (CRRA). One weakness of this model is that decision making is unaffected by scale, since the risks carried by the two securities are independent of how much wealth the investor has.

\section{Simple Problem Without Consumption}

I first present a simple system similar to the one described in {\O}kensdal [2]. Assume that the individual wishes to invest an initial amount of funds in the two-securities market. She does not plan to withdraw or deposit from this fund, thus any new asset purchased must be financed by the sale of the other asset. Her goal is to maximize the utility which is dependent only on how much the portfolio is worth at the future time $t=T$.

Let the fund's total value at time $t$ to be $Z_t$, and $u(t)$ be the Markov control that determines the fraction of wealth placed in the riskier investment at time $t$. In this case, we assume that there are no exchange costs, and since the portfolio is self financing $y(t) = 0$. Our performance function is given by $J= \mathbb{E} [U(Z_T) ] = \mathbb{E}[Z_T^\gamma/\gamma ]$, where $\mathbb{E}$ the expected value function at $t=0$.  

From Assumption 4, the wealth function satisfies the stochastic differential equation
\begin{equation}
dZ_t = dZ_t^{(u)} = Z_t(r_0(1-u(t))+r_1u(t))dt + Z_t(s_1u(t))dB_t
\end{equation}
The HJB equation using Theorem 2.1 for $\phi(s,x) = \sup_u J^u (s,x)$ is
\begin{align}
0 =&  \sup_u \{ (L^u \phi )(s,x) \}
\\ = & \frac{\partial \phi}{\partial t} + \sup_u \{ x(r_0(1-u(t))+r_1u(t))\frac{\partial \phi}{\partial x} + \frac 12 s_1^2 u^2 x^2 \frac{\partial^2 \phi}{\partial x^2} \}
\end{align}
with terminal condition
\begin{equation}
\phi(t,x) =  Z_T^\gamma/\gamma 
\end{equation}
Denote $\phi_x = \frac {\partial \phi}{\partial x}$ and $\phi_{xx} = \frac{\partial^2 \phi}{\partial x^2}$. By taking derivative of expression $(2.3)$ with respect to the control $u$, the supremum could be obtained. 
\begin{align}
0 =& x( - r_0+r_1)\phi_x + s_1^2 x^2 u \phi_{xx}\\
u =& u(t,x) = -\frac{( - r_0+r_1)\phi_x}{s_1^2 x \phi_{xx}}
\end{align}
Since $\phi_x >0$ and $\phi_{xx} < 0$, $u$ is indeed a supremum over all controls. Then, substitute the control $u$ back into the HJB equation,
\begin{align}
\phi_t +  r_0 \phi_x  -\frac{( - r_0+r_1)^2 \phi^2_x}{2s_1^2 \phi_{xx}} = 0 \qquad \text{  for }  t<T \\
\phi(t,x) = Z_T^\gamma/\gamma \qquad \text{  for }  t=T
\end{align}
The solution $\phi$ is of the form $\phi (t,x) = a(t) x^\gamma/\gamma$. By direct substitution into $(2.7)$ and $(2.8)$, and solving the simple differential equation for $a(t)$, it must be that 
\begin{equation}
\phi (t,x) = \exp \left( r_0 \gamma + \frac {(-r_0+r_1)^2 \gamma}{2 s_1^2(1-\gamma)} \right) \cdot \frac{x^{\gamma}}{\gamma}
\end{equation}
Similarly, the optimal control $u^*$ can be derived from $(4.3)$:
\begin{equation}
u^*(t,x) = \frac {-r_0+r_1}{s_1^2(1-\gamma)}
\end{equation}

\bigskip
First, we note that the optimal proportion $u^*$ is in fact constant. The growth of the portfolio $dZ_t/Z_t$ has constant mean and variance:
$$\mu = r_0(1-u^*)+r_1u^* = \frac{(r_1-r_0)^2}{s_1^2(1-\gamma)} + r_0$$
$$\sigma^2 = u^{*2}s_1^2 = \frac{(r_1-r_0)^2}{s_1^2(1-\gamma)^2}$$
Because $r_1 > r_0$ and $0 < \gamma < 1$, we have $u^* \in (0,1)$, implying there will always be investments in both securities.This is a simple example of hedging in the portfolio which demonstrates the importance of diversification. Notably, the portfolio would not require any borrowing ($u(t) \geq 1$) or shortselling ($u(t) \leq 0$) in order to maintain the optimal proportion. Since the model assumption does not allow price to drop to $0$, there would never be a risk of insolvency ($Z_t \leq 0$).

In the extraordinary case of when $r_0 > r_1$, $u$ would then be negative, Intuitively, the investor would optimally shortsell the stock for the bank asset that is safer and has high returns. However, if the shortsold stock asset yields unexpected high returns, it would be possible that the investor goes into bankruptcy.

\section{Consumption With No Transaction Costs}

Consider a more challenging two-asset problem which is similar to that presented in Merton [1]. However, we generalize the problem such that the individual also receives an influx of external income.

Assume that there are no transaction costs. At every instantaneous time, the investor is allowed to consume part of her assets (it does not matter which one since there are no transaction costs). The portfolio grows from both returns on assets and the known continuous external income. Rather than maximizing only her final worth, the investor now also wishes to take into account her utility obtained from the continuous consumption. Then the solution must determine not only the optimal amount to put in each security, but also the optimal consumption in order to maximize utility.

Let the two part control be $w = w(t, x) = (u, c) \in U$, where $u= u(t,x)$ is the fraction of the wealth invested in the risky asset, and $c = c(t,x)$ is the consumption, both at the instant $t$. As proposed in Assumption 3, let the investor's income be $y(t)$. Then, the differential wealth equation can be written as
$$ dZ_t = [Z_t(r_0(1-u(t))+r_1u(t)) + y(t) - c(t)] dt + Z_t(s_1u(t))dB_t $$
The new performance function can be described as
$$J = \mathbb{E} \left[ \int_0^T e^{- \rho t} \frac{ c(t)^{\gamma}}{\gamma} dt + B(T, X_T) \right] $$
where $\rho$ is the discount rate that determines the present value of future cash flows, and $B$ is the bequest function at exit time $T$.

Denote $\bar{c}(t) = c(t) - y(t)$. Then the problem can be reformulated as:
\begin{equation}dZ_t = dZ_t^{(u)} = [Z_t(r_0(1-u(t))+r_1u(t)) + \bar{c}(t)] dt + Z_t(s_1u(t))dB_t \end{equation}
\begin{equation} J = \mathbb{E} \left[ \int_0^T e^{- \rho t} \frac{(\bar{c}(t) + y(t))^{\gamma}}{\gamma} dt + B(T, X_T) \right] \end{equation}
under control $\bar{w} = (u, \bar{c}) \in \bar{U}$

Using the HJB equation we seek optimality with $\phi(s,x) = \underset{u,\bar{c}}{\sup} \{ J^{u,\bar{c}} (s,x) \}$:
$$ 0 =  \underset{u,\bar{c}}{\sup} \{  f^{u,\bar{c}} (t, x) + (L^{u,\bar{c}} \phi )(t,x) \}  $$
$$=  \sup_{u,\bar{c}} \left\{ e^{-\rho t} \frac{ (\bar{c} + y)^{\gamma} }{\gamma} + \frac{\partial \phi}{\partial t} + [ x(r_0(1-u)+r_1u) - \bar{c} ] \frac{\partial \phi}{\partial x} + \frac 12 s_1^2 u^2 x^2 \frac{\partial^2 \phi}{\partial x^2} \right\}$$
In order to attain the supremum condition, the two first order conditions with respect to $u$ and $\bar{c}$ are computed:
\begin{align}
\phi_{\bar{c}} = 0  \quad \Rightarrow&  \quad  \bar{c}=\bar{c}(t,x) = (\phi_x \cdot e^{\rho t} )^{1/(\gamma - 1)} - y(t)
\\ \phi_u = 0  \quad \Rightarrow&  \quad u = u(t,x) = -\frac{( - r_0+r_1)\phi_x}{s_1^2 x \phi_{xx}}
\end{align}
Substituting $u, \bar{c}$ back into the original HJB equation, 
\begin{align}
0 =& \: \phi_t + r_ 0 x \phi_x - \frac{(r_1-r_0)^2 \phi_x^2}{2s_1^2\phi_{xx}} - (\phi_x e^{\rho t})^{1/(\gamma-1)}\phi_x \\
&+ y \phi_x + \frac{e^{-\rho t}}{\gamma} (\phi_x e^{\rho t})^{\gamma /(\gamma-1)}
\end{align}
First, if $\phi = e^{-\rho t} \tilde{\phi}$, all multiples of $ e^{\rho t}$ could be canceled such that
$$0 = \tilde{\phi}_t - \rho \tilde{\phi} + r_ 0 x \tilde{\phi}_x - \frac{(r_1-r_0)^2 \phi_x'^2}{2s_1^2\tilde{\phi}_{xx}} - (\tilde{\phi}_x)^{\gamma /(\gamma-1)} + y \tilde{\phi}_x + \frac{ (\tilde{\phi}_x)^{\gamma /(\gamma-1)}}{\gamma}$$ 

Let $\tilde{\phi}$ have the solution form $\tilde{\phi} = \frac{a(t)}{\gamma}(x+N(t))^{\gamma}$. This allows the solution to clear out the extra $y(t)$ term by equating the powers of $(x+N(t))$. The $1/\gamma$ term is added so that the derivative of $\phi$ would look nice. Thus,
$$\frac{N'(t) + y(t) + r_0x}{x+N(t)}$$
must be a constant. However, since $N(t)$ is only dependent on $t$, the constant must be $r_0$, thus we must solve the simple differential equation
\begin{equation} y(t) + N'(t) = r_0N(t)\end{equation}
Evidently, if $y(t)= 0$, then $N(t) = 0$ would be a simple solution to the problem. Now, as $t_0 \rightarrow T$, since total expected future income becomes\ $0$, this is similar to the situation that there is no income, or $y(t) \rightarrow 0$ for $ t \in [t_0, T]$. Equivalently then, $N(t) \rightarrow 0$, thus $N(T) =0$. I claim that $N(t)$ is described as the investor's future income stream at time $t$ (which is added to her current income in the expression $x+N(t)$). Thus $N(t)$ can be described as,
\begin{equation} N(t) = \int_t^T y(s)e^{-r_0(s-t)} ds \end{equation}
To prove this, we first verify the boundary condition, when $t=T$, $N(T)$ is obviously $0$. Then, by using the theorem of  differentiating under the integral sign [7],
\begin{align*} N'(t) =& -y(t)e^{-r_0(t-t)} + \int_t^T r_0 y(s)e^{-r_0(s-t)} ds \\ =& -y(t) + r_0\int_t^T y(s)e^{-r_0(s-t)} ds \end{align*}
Thus, the presented expression for $N(t)$ is indeed the solution to the differential equation.  

Plugging $\phi =   e^{-\rho t}\frac{a(t)}{\gamma}(x+N(t))^{\gamma}$  back into $(5.4)$, the HJB-equation may be simplified to
$$0 = a'(t) - a(t) \rho +  a(t)r_0 \gamma - a(t)\frac{ (r_1-r_0)^2 \gamma}{2 s_1^2 (\gamma -1)}  + a(t)^ {\gamma /(\gamma-1)} (1-\gamma) $$
Denote $\mu = -\rho +  r_0 \gamma +  \frac{(r_1-r_0)^2  \gamma}{2 s_1^2 (1 - \gamma)}$. This ODE is of the form of the Bernoulli equation, which has exact solutions:
For $\mu \neq 0$:
$$a(t) = \left[ \frac{\gamma -1 + \exp \left( \frac{\mu \cdot (C_2 - t)} {1-\gamma} \right)}{\mu} \right]^{1 -\gamma}$$
and for $\mu = 0$:
$$a(t) = \left[ \frac{C_2 - t + t\gamma}{1-\gamma} \right]^{1-\gamma}$$
where $C_2$ a constant depending on the bequest function.
Now, the entire problem may be solved for in $\phi$. Consequentially, since $c(t) = \bar{c}(t) + y(t)$ the optimal controls are 
\begin{equation} c^*(t, x) = a(t)^{1/(\gamma -1)} \cdot \left(x + N(t) \right)  \end{equation} 
\begin{equation} u^*(t,x) = \frac{(r_1-r_0)(x + N(t))} {x s_1^2 (1 - \gamma )} \end{equation}
We need to check the conditions of the problem to verify that the solutions are well defined, specifically that 

1. $\phi$ is monotonically increasing and concave with respect to $x$

2. The optimal consumption function $c$ is greater than $0$.

\noindent Since $\gamma \in (0,1)$, conditions 1 and 2 are both satisfied if and only if we have $a(t) \geq 0$. Thus, we simply need for $\mu \neq 0$:
$$\left( \gamma -1 + \exp \left( \frac{\mu \cdot (C_2 - t)} {1- \gamma} \right) \right) \mu^{-1} \geq 0$$
and for $\mu = 0$
$$C_2 \geq t(1-\gamma)$$
for all values of $t \in [0,T]$.

For the sake of a complete solution, let the investor have a constant income per unit time $y(t) =y$, and that $\mu \neq 0$. Furthermore, take $B[T, Z_T] = 0$ - that is, at time $T$, the investor will not gain any further utility from any excess assets. 
Using $(5.6)$ and taking $y(t) = y$ to be a constant, $N(t)$ can be solved as
\begin{equation}N(t) = \frac y{r_0}(1 - e^{r_0(t-T)}). \end{equation}
On the other hand, the boundary conditions demand that $a(T) = 0$, or equivalently, 
$$ \exp \left( \frac{\mu \cdot (C_2 - T)} {1- \gamma} \right) = 1 - \gamma$$
Solving for $C_2$, we get
$$C_2 =(1 - \gamma )/ \mu \cdot \log(1-\gamma)  +T$$
Plugging in the expression for $a(t)$ and simplifying, we get
\begin{equation}a(t) = \left[ \frac{(\gamma-1)( 1 - \exp ( \frac { \mu (T - t)}{1-\gamma}))}{\mu} \right]^{1-\gamma} \end{equation}
which is well defined for all $\mu \neq 0$.

\noindent \emph{Demonstration:}
Since $\gamma < 1$, for $\mu>0$, well definedness holds if \begin{equation}1-\exp(\mu(T-t)/(1-\gamma))\end{equation} is negative. This is true because $\mu > 0$ and $t<T$, thus making the term $e^{\mu(T-t)/(1-\gamma)}$ greater than 1. As for $\mu<0$, well definedness holds if the term $(5.11)$ is positive. However, now since $\mu > 0$, $e^{\mu(T-t)/(1-\gamma)}$ is less than 1. This implies that for all $\mu$ and $t<T$, the expression for $a(t)$ is well defined for this problem. 
\begin{rem}
In the case of $\mu =0$, $C_2 = T(1-\gamma)$. Then, $a(t) = (T-t)^{1-\gamma} \geq 0$, which in itself is a surprisingly simple expression for the time effect. 
\end{rem}

Thus, the full complete solution for the problem with bequest function $B = 0$ would be 
\begin{equation} c^*(t, Z_t) = \frac{ \mu (Z_t r_0 + y(1 - e^{r_0(t-T)}) }{r_0(1 - \gamma)\left( \exp \left( \frac { \mu (T - t)}{1-\gamma}\right) -1 \right)} \end{equation}
\begin{equation} u^*(t, Z_t) =  \frac{(r_1-r_0)(X_t r_0 + y(1 - e^{r_0(t-T)})} {Z_t s_1^2 (1 - \gamma ) r_0} \end{equation}

\bigskip

An important observation from the complete solution $(5.12)$ and $(5.13)$ is that when $y=0$, the optimal proportion remains the same as the fixed proportion $\bar u$ derived from previous problem for which only the final wealth matters. With the addition of income, the portfolio selection is no longer independent of wealth and time. This is expected, since with a higher income stream or a greater current income, the investor is able to accept a higher risk portfolio in order to attain a larger yield. Mathematically, this can be represented as $\partial u / \partial y > 0$ and $\partial u / \partial x > 0$. On the other hand, if instead $y$ is negative, for which we can discuss as a fixed essential consumption, the portfolio holding proportion would shift more towards the safe asset.

It makes sense to denote the \emph{effective future wealth} of the individual to be the present value of the income flow subjected to the risk free interest rate, or $Y(t)$, where
\begin{equation} Y(t) =  N(t) = \int_t^T y(t)e^{-r_0(s-t)} dt \end{equation} 
The incentive to alter from the proportion $\bar u$  decreases as $t \rightarrow T$ since the effective future wealth tends to $0$. 

When computing the optimal consumption, the investor views her \emph{effective total wealth} as the sum of her current wealth and the effective future wealth. Then, this set of solutions corresponds with the properties of Friedman's Permanent Income Hypothesis [8]. If the effective total wealth was fixed, 
$$\frac{c'^*(t)}{c^*(t)}=\frac{\partial c^*(t)}{\partial t} \frac{1}{c^*(t)} = \frac{\mu}{(1-\gamma)(\exp \left( \frac { \mu (T - t)}{1-\gamma}\right)-1)}$$
denotes the instantaneous growth rate of consumption with respect to time. Evidently, there will be more consumption as $t$ goes on since $c'^*(t) /c^*(t) > 0$. However, more interestingly, we can derive an economic significance behind $\mu$. If $\mu > 0$, then $c'^*(t) /c^*(t)$ generates a graph of exponential growth form. This means that the majority of consumption growth occurs much more rapidly as time goes on. On the other hand, when $\mu < 0$, $c'^*(t) /c^*(t)$ is of the form of a reverse exponential, and the consumption growth occurs much earlier. Then, $\mu$ acts as a measurement of propensity to consume in the future, for which the benefits of high reward investment options are pitted against the the drawbacks of high discount factor.

\section{Consumption with transaction costs}

This section presents an exposition of Magill and Constantinides' paper [3] with slight modifications and simplifications. 

The previous problems show security trading in continuous time. Since there is exactly one optimal proportion for every possible state in the solvency region, the portfolios need infinitesimally small continuous adjustments in order to reach an optimal control dictated by the HJB equation. Then, the costlessness of the transaction becomes an improbable assumption, since even though trading opportunities are available continuously in time, the investor cannot execute trade continuously. Therefore, in this problem, we introduce a transaction costs such that trading does not occur continuously. It is expected that this addition will result in the investor to use her available trading opportunities at random but optimal intervals, which is a much more realistic scenario. 

The investor starts with the wealth $Z_0 = (z_0(0), z_1(0))$ where $z_0(0)$ is her holding in cash and $z_1(0)$ is the wealth in stocks. Then her two assets satisfies the stochastic differential equation
\begin{align} dz_0 &= [r_0z_0 + y(t) -  c(t) - (1 + \bigchi_0 + \bigchi_v)v(t) ]dt \\
dz_1 &= [r_1z_1 +v]dt + [s_1z_1]dB_t \end{align}
where $v(t) = z_1'(t) \cdot p_1(t) = - z_0'(t) \cdot p_0(t)$ is transaction amount at time $t$ that the investor decides upon.Assume that performance function does not change from that of the previous problem:
\begin{equation}J^{v,c} = \mathbb{E} \left[ \int_0^T e^{- \rho t} \frac{ c(t)^{\gamma}}{\gamma} dt \right]\end{equation}
Let the control be $w(t,x) = w = (c,v)$. In order to solve the problem using stochastic control problem, let $dz_1$ be modified to
\begin{equation}  dz_1= \lim_{\epsilon \rightarrow 0} [r_1z_1 +v]dt + [s_1(z_1 + \epsilon v)]dB_t \end{equation}
\begin{rem}
This limiting procedure introduced by Magill and Constantinides [3] is the crux in solving this problem. Its beauty lies in that the solution controls can enter linearly into the HJB equation but do not affect the disturbance terms.
\end{rem}
Using the HJB equation  for $\phi(s,x) = \underset{v,c}{\sup} \{ J(s,x) \}$:
\begin{align} \phi(s,x) = \sup &  \{e^{-\rho t} c^\gamma/\gamma + \phi_{0}(r_0z_0 + Y -c - (\bigchi_0 + 1+ \bigchi_v)v) \\
&+ \phi_{1}(r_1z_1 +v)+ \phi_{1,1} \frac 12 s_1^2(z_1 + \epsilon v)^2 + \phi_t\} \end{align}
with appropriate terminal conditions \begin{equation}\phi(T,z_T) = 0 \end{equation} where $\phi_{i}  = \partial \phi / \partial z_i$ and $\phi_{i,j}  = \partial^2 \phi / \partial z_i z_j$.

Taking the first order condition with respect to $c$, $v$, we must have that 
\begin{align*} 0 = & e^{-\rho t} c^{*\gamma-1} - \phi_{0}\\
0 = & -\phi_{0}(\bigchi_0 + 1+ \bigchi_{v})+\phi_{1}+   \phi_{1,1} \epsilon s_1(z_1 + \epsilon v^*)  \end{align*}
which implies that
\begin{align} c^* &= (e^{\rho t}\phi_{0})^{1/(\gamma-1)}\\
v ^* &= (1/\epsilon)\{ (\phi_{0}(\bigchi_0+1+ \bigchi_{v})-\phi_{1})/ (\epsilon s_1^2 \phi_{1,1}) - z_1 \} \end{align}
Substituting $c^*$ and $v^*$ back into the HJB equation $(6.4)$

\begin{align*}0 =& \frac{e^{-\rho t}}{\gamma}(e^{\rho t}\phi_{0})^{\gamma/(\gamma-1)} +  \phi_{1}r_1z_1 +  \phi_{0}(r_0z_0 + Y) + \phi_t- \phi_0(e^{\rho t}\phi_{0})^{1/(\gamma-1)} \\
& +(\phi_{0}(\bigchi_0 + 1+ \bigchi_{v}) - \phi_1)\frac{z_1}{\epsilon} - \frac 1{2s_1^2\epsilon^2 \phi_{1,1}}(\phi_{0}(\bigchi_0 + 1+ \bigchi_{v}) - \phi_{1})^2 \end{align*}
Inspired by the last problem, the substitution 
$$\phi = \frac{e^{-\rho t}}{\gamma} a(t) (z_0 + bz_1 + N(t))^{\gamma}$$
is taken as a guess to solve the partial differential equation for $\phi$. Such that the term $(z_0 + bz_1 + N(t))$ is taken to the same power in the equation,
$$\frac{br_1z_1+r_0z_0 + Y + N'(t) + z_1(\bigchi_0 + 1+ \bigchi_{v}-b)/ \epsilon}{z_0 + bz_1 + N(t)}$$
must be a constant. Since that the coefficients of $z_0$ should be same on both sides after cross multiplying, this constant is $r_0$. Then, by equating the rest of the terms, 
\begin{equation} b = \frac{\bigchi_0 + 1 + \bigchi_{v}}{1+ \epsilon(r_0-r_1)} \qquad r_0N(t) = y(t) + N'(t) \end{equation}
The right hand equation simply describes $N(t)$ as the effective future wealth $(5.6)$ which we demonstrated in the last section. 
$$ N(t) = \int_t^T y(s)e^{-r_0(s-t)}ds = Y(t) $$
The HJB equation can be reduced to 
$$a'(t) + (1 - \gamma)a(t)^{\gamma/(1-\gamma)} + \mu a(t) = 0$$
where $$\mu = -\rho + r_0\gamma - \frac{\gamma(\bigchi_0 + 1 + \bigchi_v - b)^2}{2b^2\epsilon^2s_1^2(\gamma-1)}=\left[-\rho + r_0\gamma + \frac{\gamma (r_1-r_0)^2}{2(1 - \gamma)s_1^2} \right]$$
such that it satisfies the terminal condition $a(T) = 0$.
Solving the differential equation for $a(t)$, the solution does not change from that of $(5.10)$,
\begin{equation} a(t) = \left[\frac{(1 - \gamma)(\exp ( \frac{\mu(T-t) - 1}{1- \gamma})}{\mu} \right] ^{1-\gamma}\end{equation}
which is well defined for all $\mu \neq 0$ (see demonstration in last section).  

Then, a full set of solution under $z=(z_0,z_1)$ can be characterized as
{\footnotesize \begin{align} \phi(t,z) =& \frac{e^{-\rho t}}{\gamma} \left[\frac{(1-\gamma)(\exp ( \frac{\mu(T-t)}{1- \gamma}- 1)}{\mu} \right] ^{1-\gamma} \left( z_0 + \frac{\bigchi_0 + 1 + \bigchi_{v}}{1+ \epsilon(r_0-r_1)}z_1 + Y(t)\right)^{\gamma} \\
c^*(t,z) =& \mu \left[z_0 + \frac{\bigchi_0 + 1 + \bigchi_{v}}{1+ \epsilon(r_0-r_1)}z_1 + Y(t) \right]\left[(1-\gamma)(\exp ( \frac{\mu(T-t)}{1- \gamma})- 1)\right]^{-1}\\
v^*(t,z) =& \frac 1{\epsilon} \left\{ \frac{(r_0-r_1)} {(\gamma-1)s_1^2} \cdot \frac{1 + \epsilon(r_0-r_1)}{\bigchi_0 + 1 + \bigchi_{v}} \left(z_0 + \frac{\bigchi_0 + 1 + \bigchi_{v}}{1+ \epsilon(r_0-r_1)}z_1 + Y(t) \right)- z_1 \right\} \end{align} }

The expression for consumption is not different from that of the second problem $(5.12)$, so the analysis will be omitted. Instead, we will look in depth at the transaction function $v^*$ and make sense of it. The formulation for $v^*$ is surprising in that as $\epsilon \rightarrow 0$, $v^* \rightarrow \infty$. This can be explained by $v^*$ being acting as an instantaneous change in order to satisfy a specific criteria, which gives rise to the desired behavior that trading occurs at random instances rather than continuously.

To make a complete analysis, we divide $v^*$ by the effective total wealth of the individual $W(t)$
$$W(t) = z_0(t) + z_1(t) + Y(t)$$
assuming that $W(t) > 0$. Denote $\pi_i = z_i/W$, $\pi_y = Y/W$, thus, $\pi_0 + \pi_1 + \pi_y =1$. Let $\pi^0 = \frac{(r_1-r_0)} {(1 -\gamma)s_1^2}$. Note that $\pi^0 > 0$, and that $\pi^0$ is once again optimal control $\bar u$ derived in the first problem.
This allows us to isolate the action of $v^*$ to only the proportion of effective total wealth in the current risky asset $\pi_1$. 
$$\lim_{\epsilon \rightarrow 0} \: \epsilon \: \frac{\bigchi_0 + 1 + \bigchi_{v}}{1 + \epsilon(r_0-r_1)} \frac{v^*}W =  \nu(\pi, \bigchi)$$
\begin{align}
\ \nu(\pi, \bigchi) &= \pi^0(\pi_0 + (\bigchi_0 + 1 + \bigchi_{v})\pi_1 + \pi_y)-\pi_1(\bigchi_0 + 1 + \bigchi_{v}) \notag \\
&=\pi^0(\pi_0 + \pi_1 + \pi_y) + \pi^0 \bigchi_0 \pi_1 + \pi^0 \bigchi_v \pi_1 -\pi_1(\bigchi_0 + 1 + \bigchi_{v}) \notag\\
&=(\bigchi_v + \bigchi_0)\pi^0 \pi_1  + \pi^0- (\bigchi_0 + 1 + \bigchi_v )\pi_1
\end{align}
The function $\nu$ acts as a signal function in that as soon as $\nu \neq 0$, $v \rightarrow \pm \infty$, and securities will be traded instantaneously. 
If $\bigchi_v = \bigchi_0 = 0$, that is, the transaction cost is $0$, the investor exhibits no trade if and only if $\pi^0 \neq \pi_1$. This corresponds exactly with $(5.13)$ from the second problem, in which we can rearrange the optimal proportion equation to get
$$\frac{z_1(t)}{Z(t)} = u^*(t,z) =  \frac{(r_1-r_0)W(t)} {s_1^2 (1 - \gamma ) r_0 Z(t)} \implies \pi^0 = \pi_1 $$
The portfolio policy will force  $\pi_1 $ to be the same as $\pi_0$ at every instant, which highlights the massive amount of trading that takes place so that this strategy holds. 

With nonzero transaction costs, the analysis becomes more interesting. The equality $\nu = 0$ occurs when 
\begin{align} 0 =& (\bigchi_v + \bigchi_0) \pi^0  + \pi^0/ \pi_1 - (\bigchi_0 + 1 + \bigchi_v )\\
\pi_ 1 =& \frac{\pi^0 }{-(\bigchi_v + \bigchi_0) \pi^0 + \bigchi_0 + 1 + \bigchi_v}\\
=& \frac{\pi^0 }{(\bigchi_v + \bigchi_0) (1 - \pi^0) +1} \end{align}
Since $\bigchi_v$ assumes different values depending of the sign of $v$, then there must be at least 2 value for which no trade will occur. Let them be
\begin{align}\frac{\pi^0 }{(\bigchi + \bigchi_0) (1- \pi^0)+ 1} =& L \\
\frac{\pi^0 }{(-\bigchi + \bigchi_0) (1 - \pi^0 ) + 1} =& H \end{align}
Evidently, $L< H$, and $L, H$ divide the real line into three regions. Then, we need to consider three cases of the location of $\pi_1$:\\
\noindent \emph{Case 1:} $\pi_1 < L \iff v^* > 0$.

If $\pi_1$ (the proportion of risky asset) is too low, the optimal portfolio policy mandates that $v^*$ is chosen such that $\pi_1$ approaches $L$  immediately as $\epsilon \rightarrow 0$. Thus, $v^* > 0$ i.e. investments will be transferred into the risky asset from the bank. Once $\pi_1 = L$ hits, $\nu = 0$ and trading ceases. On the other hand, if $v^*>0$, then $\bigchi_v = \bigchi$ and $\nu > 0$. Then by $(6.11)$, 
$$\pi_1 < \frac{\pi^0 }{-(\bigchi + \bigchi_0) \pi^0 + \bigchi_0 + 1 + \bigchi} = L$$ \\
\noindent \emph{Case 2:} $\pi_1 > H \iff v^* < 0$.

The proof here is analogous to case 1. The investor will trade immediately with $v^*<0$ until $\pi_1 = H$ hits. \\
\noindent \emph{Case 3:} $L \geq \pi_1 \geq H \iff v^* = 0$ 

Case 3 follows immediately from Case 1 and Case 2. This means that the investor will not trade the the proportion $\pi_1$ in the region $[L, H]$. 

\begin{prop}: Suppose that the Assumptions 1-5 and the performance function are satisfied. Then the investor always confines her portfolio proportion to a fixed region $K = [L,H]$, where $L,H$ are described in $(6.12)$. Any deviation of the portfolio proportion to the region will be acted upon immediately with transaction $v$ so that the portfolio returns to the nearest boundary of $K$.
\end{prop}

An excellent visual representation of Proposition 6.2 can be found in Davis and Norman's paper [9] which builds upon the work in Magill.

The portfolio proportion $\pi_1 = z_1/W$ acts as a stochastic diffusion process on region $K$ for which every time it exits the boundary of $K$, it is bought back via the optimal transaction control. This corresponds to the investor trading at random finite instances of $t \in [0,T]$, which is a reasonable trading action that we set out to model. The economic implication is that the investor needs to balance her improved diversification with costs of transaction.  When the proportion is inside the region $K$, the investor does not find it worthwhile to alter proportion so that it is closer to the optimal diversification proportion $\pi^0=\bar u$ since the benefits of diversification is underwhelmed by the transaction costs. Whereas, if the proportions exit the region $K$, the improved diversification benefits exceeds transaction costs, and it would be meaningful for the investor to trade so that $\pi_i$ goes to the boundary of $K$.

Suppose the actual trading costs from cash to stock is $\lambda$, and stock to cash is $\delta$, then $(\lambda + \delta)/2 = \bigchi$, and $(\lambda - \delta)/2 = \bigchi_0$. This indicates that $\bigchi$ is the average costs of trading, and $\bigchi_0$ is the premium of stocks over cash. If the average cost of trading $\bigchi$ goes up, the region $K$ grows as $L$ becomes smaller, and $H$ larger. A larger region indicates that the frequency of trading decreases since it is harder for the proportion processes to exit the boundary. The only case for which $L = H$, such that there is only exactly one point of optimal proportion, is when $\bigchi = \bigchi_v = 0$, i.e. the average trading costs is $0$ and the investor does not lose any money from trading back and forth the same amount at a certain instant. The primary effect of increasing the value of $\bigchi_0$ is to shift the value of both $L$ and $H$ down. This is expected since if the trading cost to get the risky asset is high compared to that of getting cash, it should be more favorable and worthwhile for the investor to hold on to the safe asset, which to her is worth dollar to dollar the same amount.

\section{Transaction Costs and a Mandatory Bequest}

In this section,we set out to solve the problem with the mandatory bequest function. First, I present a small story that serves as the backdrop as well as the motivation for this problem originally: 
\vspace{3 mm}
 
\emph{``An ancient village has 2 assets that the mayor must decide to invest in: pigs, which grows consistently based on how many pigs are owned, and gold, which has value that fluctuates depending on the market. It is currently the fall. At any time, the mayor can trade with neighboring villages for gold and pigs, however, she must pay a small transaction cost. The mayor can also direct the village to consume some of the pigs, which will give her some popularity within the town. When winter hits, the mayor must have a fixed amount of pigs for the village to store and eat in order to survive. Any extra assets left will have a minimal effect but positive on the popularity of the mayor. How can the mayor maximize her own popularity? "}
\vspace{3mm}

With this formulation of the problem, the bequest function can be written as: 
$$B[T, Z_T] =A' \cdot 1_{Z_T > K} $$
where $1$ is the indicator function, and $A'$ a suitable constant. 
Define the effective total wealth $W(t,z) = Z_t + N(t)$. Ideally, we can transform the bequest function to the form 
\begin{equation} W(T)^\gamma G(T)e^{-\rho T}. \end{equation}
Then the solution to the stochastic control problem should be readily solvable given the initial conditions and the appropriate $y(t)$ definition. Note that we assume that optimal $c(t) << K$, then by decreasing marginal utility, the extra utility derived from holding onto a wealth a significantly percentage larger than $K$ is effectively $0$. 

Let \begin{equation}N(t) = - K \cdot 1_{t = T} \qquad \text{ for } t \in [0, T] \end{equation}
This mandates that the investor have $K$ amounts of total wealth at time $t=T$, but have $0$ wealth for all $t < T$. Then by setting appropriate constants to $G(T)$, the $N(t)$ choice allows $(7.1)$ to be an alternative representation of the desired bequest function. 

In the original formulation for $N(t)$ in $(5.5)$, it is described as the solution to 
$$r_0 N(t) =  N'(t) + y(t)$$
Thus, we need to find an appropriate $y(t)$ (both mathematically and economically coherent with the model) in order to obtain the $N(t)$ within the final solution. The key idea noting that $N(t)$ can be approximated  by a curve similar to the Gaussian distribution pdf with mean $T$ and variance $\sigma^2$:
\begin{equation} N(t) = - \lim_{\sigma \rightarrow 0} K \cdot e^{-\frac{(t - T)^2}{2 \sigma^2}} \end{equation}
Then,
$$N'(t) =- \lim_{\sigma \rightarrow 0}  K \cdot \frac{(T-t)}{\sigma^2} \cdot e^{-\frac{(t - T)^2}{2\sigma^2}}$$
such that 
$$y(t) = - \lim_{\sigma \rightarrow 0} K \cdot e^{-\frac{(t - T)^2}{2 \sigma^2}} ( r_0 - \frac {T-t}{\sigma^2})$$
Since as $\sigma \rightarrow 0$, $r_0 << (T-t)/ \sigma^2$, we can ignore $r_0$
\begin{equation} y(t) = \lim_{\sigma \rightarrow 0} K \cdot \frac{(T-t)}{\sigma^2} \cdot e^{-\frac{(t - T)^2}{2 \sigma^2}} = - N'(t).\end{equation}
Note that
$$N(T) - N(a) = \int\limits_a^T N'(t) dt = - \int\limits_a^T y(t) dt$$
As $\sigma \rightarrow 0$, there exists $\epsilon(\delta) \rightarrow 0$ such that $|N(T-\epsilon) - N(0)| < \delta$, where $\delta$ is any small positive constant. Since $N(T) - N(0) = -K$, we can establish a bound
$$-(K - \delta)>N(T) - N(T-\epsilon)>- K$$
This indicates that if there is a jolt of positive income in a short period of time such that $$\int\limits_{T-\epsilon}^T y(t) \approx K,$$then the problem can be solved for the mandatory bequest function of saving wealth with value $K$. 

If instead $-K$ were to be the value of what we wish to attain, the problem statement would have the interpretation that the investor is rapidly taking out $K$ amounts of cash near the end of $T$. Any excess cash would be funneled into the bequest  $B[T,Z_T]$,which gives a minimal but positive value to the remaining cash. This gives us a solvable alternative construction of the original problem in that the mandatory bequest is handled by the fixed income function, where the income function tends to a jolt at time $T$. 

We will now formulate the updated equivalent problem explicitly.
\begin{align} dz_0 =& [r_0z_0 + y(t) - c(t) - (1 + \bigchi_v)v(t) ]dt \\
dz_1=& [r_1z_1 +v]dt + [s_1z_1]dB_t \end{align}
Subjected to the performance function
\begin{equation} J^{v,c} =  \left[ \int_0^T e^{- \rho t} \frac{ c(t)^{\gamma}}{\gamma} dt + e^{-\rho T}A'( x -  K)^{\gamma} \right] \end{equation}
and income function
\begin{equation} y(t) =  \lim_{\sigma \rightarrow 0}  -K \cdot \frac{(T-t)}{\sigma^2} \cdot e^{r_0-\frac{(t - T)^2}{2 \sigma^2}}, \qquad N(T) = 0 \end{equation}
We need to solve for the two unknown functions $W(t)$ and $a(t)$ in order to determine the HJB equation solution. Using the coefficients of $(6.7)$, the estimate of $(7.4)$, and the terminal condition of $(7.7)$, we may solve for $N(t)$ and derive 
\begin{align} W(t) = x + N(t) =& \lim_{\sigma \rightarrow 0} \: z_0 + \frac{\bigchi_0 + 1 + \bigchi_{v}}{1+ \epsilon(r_0-r_1)}z_1 + (K \cdot e^{-\frac{(T-t)^2}{2 \sigma^2}} - K) \\
 =& z_0 + \frac{\bigchi_0 + 1 + \bigchi_{v}}{1+ \epsilon(r_0-r_1)}z_1 - K \cdot (1 - 1_{t = T}) \end{align}
Since the Bernoulli equation remains unchanged from that of section 5, $a(t)$ does not change: 
$$a(t) = \left[ \frac{\gamma -1 + \exp \left( \frac{\mu \cdot (C_2 - t)} {1-\gamma} \right)}{\mu} \right]^{1 -\gamma}$$
However, this time the bequest function mandates that $a(T) = A'$. Solving for $C_2$ and then $a(t)$
$$C_2 =(1 - \gamma )/ \mu \cdot \log(\mu A'^{1/(1-\gamma)} + 1-\gamma)  +T$$
\begin{equation} a(t) = \left[ \frac{\gamma - 1 + (\mu A^{1/(1-\gamma)} + 1 - \gamma)\exp ( \frac { \mu (T - t)}{1-\gamma})}{\mu} \right]^{1-\gamma}.\end{equation}
Optimal controls $u$ and $c$ then are respectively:
\begin{align} c^*(t,z) =& \mu W(t) \left[\gamma - 1 + (\mu A^{1/(1-\gamma)} + 1 - \gamma)\exp ( \frac { \mu (T - t)}{1-\gamma})\right]^{-1}\\
v^*(t,z) =& \frac 1{\epsilon} \left\{ \frac{(r_0-r_1)} {(\gamma-1)s_1^2} \cdot \frac{1 + \epsilon(r_0-r_1)}{\bigchi_0 + 1 + \bigchi_{v}} W(t)- z_1 \right\} \end{align}

Since the problem follows closely the permanent income hypothesis [8], it is expected that the investor will consume less when she has a mandatory bequest function. Indeed, while taking into account of her effective total wealth, the investor will subtract from her current wealth the amount that she needs to pay at $T$. The constant $A'$ measures how important the excess cash is, and  $\partial c^* / \partial A'$ depends on the sign of $\mu$, which describes whether it is better to consume now or in the future. An curious case is when $A'$ drops below $0$ - that is, the investor is penalized if she keeps a certain amount that exceeds the mandatory payment. When $\mu < 0$, the consumption may even lower because the discount factor will outweigh the investment opportunities. Of course, we still must beware of the initial well-definedness condition that $a(t)>0$. 

The transaction control $v^*$ on the other hand does not take into account the bequest function at all. The investor will choose a transaction policy that will maximize the expected utility when subjected to her effective total wealth. It is important to note that when the mandatory bequest value is increased, $W(t)$ increases and $\pi_1$ decreases while the benchmark $\pi^0$ is unchanged. From Proposition 6.2, the optimal proportions for the risky asset drops, which is accompanied by a corresponding rise in the optimal proportions for the safe asset. Intuitively, the investor is taking on a safer path so that her satisfying of the mandatory bequest is more assured. A detailed analysis regarding the transaction behaviors due to the transaction costs has been provided in the previous section.

\section{Conclusion}

This paper uses the powerful HJB equation to systematically construct and analyze optimal continuous time dynamic portfolio models via many concrete examples. The basic methods presented here are applicable to a wide range of economic models that use underlying decision theory. The paper discusses the qualitative effects of an optimal consumption portfolio problem subjected to transaction costs and a mandatory bequest goal. The most basic change is that the investor removes the mandatory bequest from her total wealth when making decisions and confines the optimal portfolio proportions to a fixed region. A direct consequence is that the investor will consume less and will adjust her portfolio proportions at random instances tending towards the safer asset, which reflects her level of caution in order to satisfy the bequest goal. The magnitude of these adjustments correspond directly with the size of the mandatory bequest.

The models presented in this paper should be easily extended to the general $k$-asset case, so that the theory could be applied to more diverse and intercorrelated assets. Furthermore, a more complex utility function could be implemented such that the model may be more dynamic. Even with these improvements to the model, it is expected that the properties observed in this paper will continue to hold. A wider significance of this paper is to scratch the surface and elicit the studying of more elaborate and vibrant bequest functions so that they may serve a variety of economic purposes.

\end{document}